\documentclass{article}
\usepackage{amsmath,amssymb,amsthm}
\usepackage[all]{xy}

\voffset = - 1.2cm
\theoremstyle{plain}
\newtheorem{theorem}{Theorem}
\newtheorem{corollary}[theorem]{Corollary}
\newtheorem{proposition}[theorem]{Proposition}
\newtheorem{lemma}[theorem]{Lemma}

\theoremstyle{definition}
\newtheorem{definition}[theorem]{Definition}
\newtheorem{note}[theorem]{Note}

\numberwithin{equation}{section}
\numberwithin{theorem}{section}

\numberwithin{equation}{section}
\numberwithin{theorem}{section}

\begin{document}

\begin{center}
{\bf \Large {Lie--Hamilton systems: theory and applications}}
\end{center}
\begin{center}
J.F. Cari\~nena, J. de Lucas, C. Sard\'on
\end{center}

\begin{center}
Department of Theoretical Physics and IUMA, University of Zaragoza,\\ Pedro Cerbuna 12, 50.009, Zaragoza, Spain.
\end{center}
\begin{center}
Faculty of Mathematics and Natural Sciences, Cardinal Stefan Wyszy\'nski University,\\ 
W\'oycickiego 1/3, 01-938, Warszawa, Poland. 
\end{center}
\begin{center}
Department of Fundamental Physics, University of Salamanca, \\Plaza de la Merced s/n, 30.008, Salamanca, Spain.
\end{center}


\date{Received: ---- / Accepted: ----}

\begin{abstract}This work concerns the definition and analysis of a new class of Lie systems on
Poisson manifolds enjoying rich geometric 
features: the Lie--Hamilton systems. We devise methods to study their
superposition rules, time independent constants of 
motion and Lie symmetries, linearisability conditions, etc. Our results are
illustrated by examples of physical and mathematical interest.
\end{abstract}

\section{Introduction}
The use of geometric tools for studying differential equations has been proved
to be a very successful approach, as witnessed by the many
 works devoted to this topic over the years \cite{CRC94,A89,Be84,Ol91,WTC83}. 
Among these methods, we here focus on the theory of Lie systems
\cite{CGM00,CGM07,Ib99,LS,PW}. 

{\it Lie systems}  form a family of systems of first-order ordinary differential
equations  whose general solutions can be written in terms of finite
  families of particular solutions and a set of constants by a particular type
of functions, the so-called {\it superposition rules} \cite{CGM00,CGM07,LS,PW}.
   Moreover, Lie systems enjoy many geometrical properties
\cite{CGM00,CGM07,CarRamGra,FLV10,Ib99,PW}. 
 
In modern geometric terms, the {\it Lie--Scheffers Theorem} \cite{CGM07} states that a Lie
system amounts to a $t$-dependent vector field taking 
values in a finite-dimensional Lie algebra of vector fields, the so-called {\it
Vessiot--Guldberg Lie algebra} \cite{Dissertations,Ib00,Ib09,RA}. This condition is so
restrictive that only few differential equations can be considered as Lie
systems. Nevertheless, Lie systems appear in very important physical and
mathematical problems \cite{ADR12,CL11Sec,CL11Sec2,Dissertations,CLR08,
Clem06,Ru08,Ru10,Ib00,Ib09,WintSecond,RA}, which strongly
motivates their analysis.

The first aim of this work is to uncover an interesting geometric feature shared
by several Lie systems. More specifically, we show that second-order Riccati
equations \cite{CRS05,FL66I,GGL08,CGK08,GL99}, second-order 
Kummer--Schwarz equations, as well as
Smorodinsky--Winternitz oscillators (among other other remarkable examples)
 can be described by Lie systems associated to Vessiot--Guldberg Lie algebras of
Hamiltonian vector fields (with
  respect to a certain Poisson structure). In this way, we highlight that this
property deserves a thorough study, which is the main aim of the paper. 

Previous examples suggest us the definition and analysis of a new type of Lie
systems, the hereafter called {\it Lie--Hamilton systems},
 admitting a plethora of geometric properties. For instance, their dynamics is
governed by curves in finite-dimensional Lie algebras 
 of functions (with respect to a Poisson structure). These geometrical objects,
the hereafter named {\it Lie--Hamiltonian structures}, 
 are a key to understand the characteristics of Lie--Hamilton systems. 

Our achievements are employed to study superposition rules, Lie symmetries,
constants 
of motion, and other features of Lie--Hamilton systems. It is noticeable that
Lie--Hamiltonian structures allow
 us to use Poisson and symplectic geometric techniques to study Lie--Hamilton
systems. Among other achievements,
  we prove that $t$-independent constants of motion of Lie--Hamilton systems
form a {\it function group} \cite{We83}, provide
   conditions for simultaneous linearisation of Lie--Hamilton systems and their
related Poisson bivectors, and we describe properties 
   of Lie symmetries of Lie--Hamilton systems. 

All our achievements are exemplified by the analysis of Lie systems of physical
and mathematical relevance. 
Furthermore, several new concepts related to $t$-dependent vector fields are
introduced and briefly investigated as a
tool to investigate Lie--Hamilton systems.

The structure of the paper goes as follows. Section 2 concerns the description
of the notions and conventions about Poisson geometry and Lie algebras to be used
 throughout our paper. Section 3 is devoted to some concepts
of the theory of $t$-dependent vector
  fields and Lie systems. In Section 4 the analysis of several remarkable Lie
systems on Poisson manifolds leads us to
   introduce the concept of a Lie--Hamilton system, which encompasses such systems as
particular cases. Subsequently, the Lie--Hamiltonian
    structures are introduced and analysed in Section 5. Next, we investigate
several geometric properties of Lie--Hamilton in 
    Section 6. Finally, Section 7 summarises our main results and present an
outlook of our future research on these systems. 
    
\section{Fundamentals}\label{LSLS}

For simplicity, we hereafter assume all mathematical objects to be real, smooth,
and globally defined. This permits us to omit several
minor technical problems so as to highlight the main aspects of our results.

Let us denote Lie algebras by pairs $(V,[\cdot,\cdot])$, where $V$ stands for a
real linear space endowed with a Lie
 bracket $[\cdot\,,\cdot]:V\times V\rightarrow V$. Given two subsets
$\mathcal{A}, \mathcal{B}\subset V$, we write $[\mathcal{A},\mathcal{B}]$ for the real linear space spanned by the Lie brackets between elements 
 of $\mathcal{A}$ and $\mathcal{B}$, and we define ${\rm
Lie}(\mathcal{B},V,[\cdot,\cdot])$ to be  the smallest Lie subalgebra
  of $V$ containing $\mathcal{B}$. 
Note that ${\rm Lie}(\mathcal{B},V,[\cdot,\cdot])$ is a well-defined object (it
exists and is unique), which is generated by the elements of
\begin{equation}\label{family}
\mathcal{B},[\mathcal{B},\mathcal{B}],[\mathcal{B},[\mathcal{B},\mathcal{B}]],[
\mathcal{B},[\mathcal{B},[\mathcal{B},\mathcal{B}]]],[[\mathcal{B},\mathcal{B}],[\mathcal{B},\mathcal{B}]],\ldots
\end{equation}
From now on, we use ${\rm Lie}(\mathcal{B})$ and $V$ to represent ${\rm
Lie}(\mathcal{B},V,[\cdot,\cdot])$ and $(V,[\cdot,\cdot])$, 
correspondingly, when their meaning is clear
 from context.

Given a fibre vector bundle ${\rm pr}:P\rightarrow N$, we denote by $\Gamma({\rm pr})$
the $C^\infty(N)$-module of 
its smooth sections.  So, if $\tau_N:TN\rightarrow N$ and $\pi_N:
T^*N\rightarrow N$ are the canonical projections
 associated with the tangent and cotangent bundle to $N$, respectively, then
$\Gamma(\tau_N)$ and $\Gamma(\pi_N)$
  designate the $C^\infty(N)$-modules of vector fields and one-forms on
$\mathbb{R}^n$, correspondingly.

We call {\it generalised distribution} $\mathcal{D}$ on a differentiable
manifold $N$ a function that sends each $x\in N$ to  a linear
 subspace $\mathcal{D}_x\subset T_xN$. A generalised distribution is said to be
{\it regular at } $x'\in N$ when the function 
  ${\rm r}:N\to \mathbb{N}\cup \{0\}$ of the form 
 ${\rm r}:x\in N\mapsto \dim\mathcal{D}_x\in \mathbb{N}\cup\{0\}$  is locally constant
around $x'$. Similarly, $\mathcal{D}$ is regular on an open $U\subset N$ when ${\rm r}$ is constant on $U$. Finally,
a vector field  $Y\in\Gamma(\tau_N)$ is said to
   take values in $\mathcal{D}$, in short $Y\in\mathcal{D}$, when
$Y_x\in\mathcal{D}_x$ for all $x\in N$. Likewise, similar 
   notions can be defined for a {\it generalised codistribution}, namely a
mapping relating every $x\in N$ to a linear subspace of $T_x^*N$.

In what follows, a {\it Poisson algebra} $(\mathfrak{W},\star,\{\cdot,\cdot\})$ is a triple consisting of an $\mathbb{R}$-vectorial space $\mathfrak{W}$ and two bilinear maps on $\mathfrak{W}$, namely $\star$ and $\{\cdot,\cdot\}$, such that
$\mathfrak{W}$ endowed with $\star$ becomes a commutative and associative $\mathbb{R}$-algebra and $(\mathfrak{W},\{\cdot,\cdot\})$ is a Lie algebra whose  Lie bracket, the
so-called {\it Poisson bracket} of the Poisson algebra, satisfies the {\it Leibnitz rule} relative to $\star$, namely
$$
\{f\star g,h\}=f\star \{g,h\}+\{f,h\}\star g,\qquad \forall f,g,h \in \mathfrak{W}.
$$
In other words, $\{\cdot,h\}$ is a derivation of the $\mathbb{R}$-algebra $\mathfrak{W}$ for each $h\in\mathfrak{W}$.

A {\it Poisson manifold} is a pair $(N,\{\cdot,\cdot\})$ such that
$(C^\infty(N),\cdot,\{\cdot,\cdot\})$ becomes a Poisson algebra with respect to
the 
standard product ``$\cdot$'' of functions on $N$. The map $\{\cdot,\cdot\}$ is
called the {\it Poisson structure} of the Poisson manifold. Observe that
a Poisson structure is a derivation in each entry, which,
as shown next, has relevant consequences. 

On one hand, given an $f\in C^\infty(N)$, there exists a single vector field
$X_f$ on $N$, the referred to as {\it Hamiltonian vector field}
 associated with $f$, 
such that $X_fg=\{g,f\}$ for all $g\in C^{\infty}(N)$. The Jacobi identity  for
the Poisson structure therefore entails 
$$
X_{\{f,g\}}=-[X_f,X_g],\qquad \forall f,g\in C^{\infty}(N).
$$
In other words, the mapping $f\mapsto X_f$ is a Lie algebra anti-homomorphism
between the Lie algebras 
$(C^{\infty}(N),\{\cdot,\cdot\})$ and $(\Gamma(\tau_N),[\cdot,\cdot])$.  

On the other hand, a Poisson structure determines a unique bivector field $\Lambda\in 
\Gamma(\bigwedge\,^2 TN)$ such that
\begin{equation}\label{ForPoi}
\{f,g\}=\Lambda(df,dg),\qquad \forall f,g\in C^{\infty}(N).
\end{equation}
We call $\Lambda$ the {\it Poisson bivector} of the Poisson manifold
$(N,\{\cdot,\cdot\})$.
In view of the Jacobi identity for the Poisson structure, it follows that
$[\Lambda,\Lambda]_S=0$, 
with $[\cdot,\cdot]_S$ being the {\it Schouten--Nijenhuis Lie bracket}
\cite{IV}. Conversely, every bivector field
 $\Lambda$ on $N$ satisfying the  previous Schouten--Nijenhuis Lie bracket
vanishing condition
 gives rise to a Poisson structure given by (\ref{ForPoi}). 
This justifies referring to Poisson manifolds as $(N,\{\cdot,\cdot\})$ or
$(N,\Lambda)$ indistinctly. In some cases, we shall write
$\{\cdot,\cdot\}_\Lambda$ for the Poisson structure induced by a Poisson
bivector $\Lambda$ if this is may not be clear from context.

Every Poisson bivector induces a unique bundle morphism
$\widehat\Lambda:T^*N\rightarrow TN$ such that 
 $\omega'(\widehat\Lambda(\omega))=\Lambda(\omega,\omega')$ for every
$\omega,\omega'\in \Gamma(\pi_N)$. 
 This morphism allows us to relate every function $f\in C^\infty(N)$ to its
associated vector field $X_f$ through
  the relation $X_f=-\widehat \Lambda(df)$.  We define ${\rm
Ham}(N,\Lambda)$ to be the $\mathbb{R}$-linear space of Hamiltonian vector fields on
$N$ relative to $\Lambda$. 
 This space induces an integrable generalised distribution $\mathcal{F}^\Lambda$ on $N$, 
  the so-called {\it characteristic distribution} associated to $\Lambda$, of the form
  $\mathcal{F}^\Lambda_x=\{X_x\mid X \in {\rm Im}\, \widehat \Lambda\}\subset
T_xN$, with $x\in N$,  whose
leaves are symplectic manifolds
   with respect to the restrictions of $\Lambda$ \cite{We83}.

If $X_f=0$, we say that $f$ is a {\it Casimir function}. We denote by ${\rm
Cas}(N,\Lambda)$ the $\mathbb{R}$-linear space of 
Casimir functions on $N$ relative to the Poisson bivector $\Lambda$.
Finally, let us define a last structure that will be of interest in our work.
\begin{definition}
We call {\it Casimir co-distribution} of the Poisson manifold $(N,\Lambda)$ the
generalised co-distribution of the 
form $\mathcal{C}^\Lambda={\rm ker}\,\widehat \Lambda$.
\end{definition}

It is well known that the cotangent bundle of a Poisson manifold $(N,\Lambda)$
admits a Lie 
algebroid structure $(T^*N,[\cdot\,,\cdot]_\Lambda,\widehat \Lambda)$, with
anchor $\widehat{\Lambda}$ and 
Lie bracket
$[\omega,\omega']_\Lambda=\mathcal{L}_{\widehat\Lambda(\omega)}\omega'-
\mathcal{L}_{\widehat\Lambda(\omega')}\omega- d\Lambda(\omega,\omega')$, where
$\mathcal{L}_X$ denotes 
the Lie derivative with respect to a vector field $X$. 
In particular, $[df,dg]_\Lambda=d\{f,g\}$, for all $f,g\in C^{\infty}(N)$
 (see \cite{Marle08,Vo04} for further details). 

\begin{proposition} The Casimir co-distribution of a Poisson manifold is
involutive. 
\end{proposition}
\begin{proof} Given two sections $\omega,\omega'\in\mathcal{C}^\Lambda$, we have
that
 $\widehat \Lambda(\omega)=\widehat\Lambda(\omega')=0$ and then
$\Lambda(\omega,\omega')=0$. In consequence, 
$$
[\omega,\omega']_\Lambda=\mathcal{L}_{\widehat\Lambda(\omega)}\omega'-
\mathcal{L}_{\widehat\Lambda(\omega')}\omega-d\Lambda(\omega,\omega')=0.
$$
\end{proof}

\section{Time-dependent vector fields and Lie systems}

A {\it $t$-dependent vector field} on $N$ is a map $X:(t,x)\in\mathbb{R}\times
N\mapsto X(t,x)\in TN$ 
such that $\tau_N\circ X=\pi_2$, where $\pi_2:(t,x)\in\mathbb{R}\times N\mapsto
x\in N$. This condition entails that every $t$-dependent vector field amounts to
a family of vector fields $\{X_t\}_{t\in\mathbb{R}}$, with $X_t:x\in N\mapsto
X(t,x)\in TN$ for all $t\in\mathbb{R}$ and vice versa \cite{Dissertations}.

We call {\it integral curves} of $X$ the integral curves 
$\gamma:\mathbb{R}\mapsto \mathbb{R}\times N$ of the {\it suspension} of $X$,
i.e. the vector field $X(t,x)+\partial/\partial t$ on $\mathbb{R}\times N$ \cite{FM}. Every
integral curve $\gamma$ admits a parametrization in terms of a parameter $\bar
t$ such that 
$$
\frac{d(\pi_2 \circ \gamma)}{d\bar t}(\bar t)=(X\circ \gamma)(\bar t).
$$
This system is referred to as the {\it associated system} of $X$. Conversely,
every system of first-order differential equations in normal form describes the
integral curves of a  unique $t$-dependent vector field. This establishes a
bijection between $t$-dependent vector fields and systems of first-order
differential equations in normal form, which
 justifies to use $X$ to denote both a $t$-dependent vector field and its
associated system.

\begin{definition}  The  {\it minimal Lie algebra} of a $t$-dependent vector
field $X$ on $N$ is the smallest real Lie algebra, 
$V^X$, containing the vector fields $\{X_t\}_{t\in\mathbb{R}}$, namely $V^X={\rm
Lie}(\{X_t\}_{t\in\mathbb{R}})$. 
\end{definition}

Minimal Lie algebras enable us to define the following new geometric structures
that will be of interest so as to study the geometric 
properties of Lie systems, in general, and Lie--Hamilton systems, in particular.
 
\begin{definition} Given a $t$-dependent vector field $X$ on $N$, its {\it
associated distribution}, 
$\mathcal{D}^X,$ is the generalised distribution on $N$ spanned by the vector
fields of $V^X$, i.e.  
$$
\mathcal{D}^X_x=\{Y_x\mid Y\in V^X\}\subset T_xN,
$$
and its {\it associated co-distribution}, $\mathcal{V}^X$, is the generalised
co-distribution on $N$ of the form
$$
\mathcal{V}^X_x=\{\vartheta\in T_x^*N\mid \vartheta(Z_x)=0,\forall
\,\,Z_x\in \mathcal{D}_x^X\}=(\mathcal{D}^X_x)^\circ\subset T_x^*N,
$$  
where $(\mathcal{D}^X_x)^\circ$ is the {\it annihilator} of $\mathcal{D}_x^X$. 
\end{definition}

Observe that the function  ${\rm r}^X:x\in N\mapsto 
\dim\mathcal{D}^X_x\in\mathbb{N}\cup \{0\}$ needs not be constant on $N$. We can
 only guarantee that  ${\rm r}^X(x)=k$ implies  ${\rm r}^X(x')\geq  {\rm
r}^X(x)$ for $x'$ in a neighbourhood of $x$. Indeed, in this case there exist
$k$ 
 vector fields $Y_1,\ldots,Y_k\in V^X$ such that $(Y_1)_x,\ldots,(Y_k)_x\in
T_xN$ are linearly independent. As we assume vector fields 
 to be smooth, $(Y_1)_{x'},\ldots,(Y_k)_{x'}\in T_{x'}N$ are also linearly
independent for $x'$ in a neighbourhood of $x$ and hence 
  ${\rm r}^X(x')\geq k$. From here, it easily follows that ${\rm r}^X$ is a {\it
lower semicontinuous function} and must be 
 constant on the connected components of an open and dense subset $U^X$ of $N$
(cf. \cite[p. 19]{IV}), where $\mathcal{D}^X$
  becomes a regular  distribution. As for every $x\in U^X$ there exists a local
basis for $\mathcal{D}^X$ consisting of ${\rm r}^X(x)$ elements 
  belonging to $V^X$, the generalised distribution $\mathcal{D}^X$ is involutive
and integrable on each connected component of $U^X$. Since
   $\dim \mathcal{V}^X_x=\dim\, N- {\rm r}^X(x)$, then $\mathcal{V}^X$ becomes a
regular co-distribution on each component also. 

The most relevant instance for us is when $\mathcal{D}^X$ is determined by a
finite-dimensional $V^X$ and hence $\mathcal{D}^X$ 
becomes integrable on the whole $N$ \cite[p. 63]{JPOT}. 
It is worth noting that even in this case, $\mathcal{V}^X$ does not need to be a {\it
differentiable distribution}, i.e. given 
$\vartheta\in\mathcal{V}^X_x$, it does not generally exist a 
locally defined one-form $\omega\in\mathcal{V}^X$ such that $\omega_x=\vartheta$.

Let us describe a first result that justifies the definition of the above
geometric notions.

\begin{proposition}\label{NuX} A function $f:U\rightarrow \mathbb{R}$
is a local $t$-independent constant of  
motion for a system $X$ if and only if $df\in \mathcal V^X|_U$. 
\end{proposition}
\begin{proof} Under the above assumptions, $X_tf|_U=df(X_t)|_U=0$ for all
$t\in\mathbb{R}$. Consequently, $df$ also vanishes on the successive Lie
brackets of elements from $\{X_t\}_{t\in\mathbb{R}}$ and  hence
$$df(Y)|_U=Yf|_U=0,\qquad \forall Y\in {\rm Lie}(\{X_t\}_{t\in\mathbb{R}}).$$
Since the elements of $V^X$ span the generalised distribution
  $\mathcal{D}^X$, then $df_x(Z_x)=0$ for all $x\in U$ and $Z_x\in
\mathcal{D}^X_x$, i.e. $df\in \mathcal{V}^X|_U$. The converse directly follows
from the above considerations. 
\end{proof}

In brief, Proposition \ref{NuX} shows that (locally defined) $t$-independent
constants of motion of $t$-dependent vector fields are determined
 by (locally defined) exact one-forms taking values in its associated
co-distribution. Then, $\mathcal{V}^X$ is what really matters in the
 calculation of such constants of motion for a system $X$. 

Let us enunciate a lemma that will be used throughout our work and whose proof
is straightforward.

\begin{lemma}\label{basisVX} Given a system $X$, its associated co-distribution
$\mathcal{V}^X$ admits a local basis around
 every $x\in U^X$ of the form $df_1,\ldots,df_{p(x)}$, with $p(x)= {\rm r}^X(x)$
and $f_1,\ldots,f_{p(x)}:U\subset U^X\rightarrow\mathbb{R}$ 
 being a family of (local) $t$-independent constants of motion for $X$.
Furthermore, the $\mathbb{R}$-linear space $\mathcal{I}^X|_U$ of $t$-independent
constants of motion of $X$ on $U$ can be written as
$$
\mathcal{I}^X|_U=\{g\in C^\infty(U)\mid \exists F:U\subset
\mathbb{R}^{p(x)}\rightarrow\mathbb{R}, \ g=F(f_1,\ldots, f_{p(x)})\}.
$$
\end{lemma}
 \begin{note} Roughly speaking, the above lemma shows that $\mathcal{V}^X$ is differentiable on $U^X$.   
 \end{note}

Let us now turn to some fundamental notions appearing in the theory of Lie
systems. 

\begin{definition} A {\it superposition rule} depending on $m$ particular
solutions for a system $X$ on $N$ 
 is a function $\Phi:N^{m}\times N\rightarrow
N$, $x=\Phi(x_{(1)}, \ldots,x_{(m)};\lambda)$, such that the general
 solution $x(t)$ of $X$ can be brought into the form
  $x(t)=\Phi(x_{(1)}(t), \ldots,x_{(m)}(t);\lambda),$
where $x_{(1)}(t),\ldots,x_{(m)}(t)$ is any generic family of
particular solutions and $\lambda$ is a point of $N$ to be related to initial
conditions. 
 \end{definition}

The conditions ensuring that a system $X$ possesses a superposition rule are
stated
 by the {\it Lie--Scheffers Theorem} \cite[Theorem 44]{LS}. A modern statement
of this relevant result is described
 next (for a modern geometric description see \cite[Theorem 1]{CGM07}). 

\begin{theorem} A system $X$ admits a superposition rule if and only if $X$ can be written as
 $X_t={{\sum_{\alpha=1}^r}}b_\alpha(t)X_\alpha$ 
for a certain family $b_1(t),\ldots,b_r(t)$  of $t$-dependent functions and a
collection  $X_1,\ldots,X_r$ of vector fields  spanning 
an $r$-dimensional real Lie algebra.
\end{theorem}
 
Systems of first-order differential equations possessing a superposition rule
are called {\it Lie systems}. The Lie--Scheffers 
Theorem yields that every Lie system $X$ is related to (at least) one
finite-dimensional real Lie algebra of vector fields 
$V$, the so-called {\it Vessiot--Guldberg Lie algebra}, satisfying that
$\{X_t\}_{t\in\mathbb{R}}\subset V$. This implies 
that $V^X$ must be finite-dimensional. Conversely, if $V^X$ is
finite-dimensional, 
this Lie algebra can be chosen as a Vessiot--Guldberg Lie algebra for $X$. This
proves the following
 theorem, which motivates, among other reasons, the definition of $V^X$ 
\cite{Dissertations}. 

\begin{theorem}\label{ALST}{\bf (The abbreviated Lie--Scheffers Theorem)} A
system $X$ admits a superposition rule if and only if $V^X$ is
finite-dimensional.
\end{theorem}

The Lie--Scheffers Theorem may be used to reduce the integration  of a Lie
system to solving a special type 
of Lie systems on a Lie group. More precisely, every Lie system $X$ on a
manifold $N$ possessing a Vessiot--Guldberg
 Lie algebra $V$, let us say
$X_t={{\sum_{\alpha=1}^r}}b_\alpha(t)X_\alpha$, where
$X_1\ldots,X_r$ is a basis of $V$,
  can be associated with a (generally local) Lie group action $\varphi:G\times
N\rightarrow N$ whose fundamental
   vector fields coincide with those of $V$ \cite[Theorem XI]{Palais}. This
action allows us to bring the general solution $x(t)$ of $X$ into the 
   form $x(t)=\varphi(g(t),x_0)$, where $x_0\in N$ and $g(t)$ is the solution
with $g(0)=e$ of the Lie system
\begin{equation}\label{EquLie}
\frac{dg}{dt}=-\sum_{\alpha=1}^rb_\alpha(t)X_\alpha^R(g),\qquad g\in G,
\end{equation}
where $X^R_1,\ldots,X^R_r$ are a family of right-invariant vector fields on $G$
admitting the same structure constants as $-X_1,\ldots,-X_r$ (see \cite{CGM00}
for details). In this way, the explicit integration of a Lie system $X$ reduces
to finding one particular solution of (\ref{EquLie}) if $\varphi$ is explicitly
known. Conversely, the general solution of $X$ enables us to construct the
solution
 for (\ref{EquLie}) with $g(0)=e$ by solving an algebraic system of equations, provided the explicit form of $\varphi$ is given
\cite{AHW81}.

\section{Lie--Hamilton systems}
A few instances of Lie systems on Poisson manifolds have recently appeared
during the analysis of various mathematical and physical problems
\cite{ADR12,CGM00,CLS11,Ru10}. In all these cases, and several new ones to be
presented here, the structure of the Lie system can be related to the Poisson
manifold in a special way. Let us analyse this question in depth to motivate our
definition of a Lie--Hamilton system.

Consider a second-order Riccati equation, i.e. a second-order differential
equation of the form
\begin{equation}\label{NLe}
\frac{d^2x}{dt^2}+(g_0(t)+3g_1(t)x)\frac{dx}{dt}
+c_0(t)+c_1(t)x+c_2(t)x^2+c_3(t)x^3=0,
\end{equation}
where
$$
g_1(t)=\pm \sqrt{c_3(t)},\qquad
g_0(t)=\frac{c_2(t)}{g_1(t)}-\frac{1}{2c_3(t)}\frac{dc_3}{dt}(t), \qquad
c_3(t)> 0,
$$
which appears in the study of interesting physical and mathematical problems
\cite{CL11Sec,CL11Sec2,CLS11,CRS05,CC87,GGL08,GL99}. 

Recently, it was found that a very general family of second-order Riccati
equations admits a Lagrangian description in terms of a $t$-dependent
non-natural regular Lagrangian
$$
L(t,x,v)=\frac{1}{v+U(t,x)},
$$
where $U(t,x)=a_0(t)+a_1(t)x+a_2(t)x^2$ and $a_0(t),a_1(t),a_2(t)$ are certain
functions related to the
 $t$-dependent coefficients of (\ref{NLe}) \cite{CRS05}. 

The Legendre transformation induced by the above Lagrangian leads to $$
p=\frac{\partial  L}{\partial v}=-\frac{1}{(v+U(t,x))^2}\Longrightarrow v=\pm\frac
1{\sqrt{-p}}-U(t,x),
$$
and hence the image of the Legendre transformation is the open submanifold $\mathbb{R}\times\mathcal{O}$, where $\mathcal{O}\equiv\{(x,p)\in {\rm
T}^*\mathbb{R}\mid p< 0\}$ (see \cite{CLS11} for details). If we restrict to the points $(t,x,v)$ where $v+U(t,x)>0$ (assuming the contrary leads to similar results), we can define in $\mathbb{R}\times\mathcal{O}$ the 
$t$-dependent Hamiltonian
function
$$
h(t,x,p)=vp-L(t,x,v)=p\left(\frac 1{\sqrt{-p}}-U(t,x)\right)-\sqrt{-p}=-2\sqrt{-p}- p\,
U(t,x).
$$
Therefore, the Legendre transformation maps second-order Riccati equations
(written as a first-order system) into the $t$-dependent Hamilton equations on
$\mathcal{O}$ \cite{CLS11}:
\begin{equation}\label{FORiccati}
\left\{
\begin{aligned}
\frac{dx}{dt}&=\frac{\partial h}{\partial
p}=\frac{1}{\sqrt{-p}}-a_0(t)-a_1(t)x-a_2(t)x^2,\\
\frac{dp}{dt}&=-\frac{\partial h}{\partial x}= p(a_1(t)+2a_2(t)x),
\end{aligned}\right.
\end{equation}
 The above system is a Lie system as it describes the integral curves of the
$t$-dependent vector field 
$$
X_t=X_1-a_0(t)X_2-a_1(t)X_3-a_2(t)X_4
,$$
where
$$
\begin{aligned}
X_1=\frac{1}{\sqrt{-p}}\frac{\partial}{\partial x},\quad
X_2=\frac{\partial}{\partial x},\quad
X_3=x\frac{\partial}{\partial x}-p\frac{\partial}{\partial p},\quad
X_4=x^2\frac{\partial}{\partial x}-2xp\frac{\partial}{\partial p},
\end{aligned}
$$
along with
$$
X_5=\frac{x}{\sqrt{-p}}\frac{\partial}{\partial
x}+2\sqrt{-p}\frac{\partial}{\partial p},
$$
span a five-dimensional Lie algebra of vector fields. In addition, this Lie
algebra enjoys an additional property that has not been noticed so far: all
their elements are Hamiltonian vector fields with respect to the Poisson
bivector $\Lambda=\partial/\partial x\wedge\partial/\partial p$ on
$\mathcal{O}$. Indeed, note that $X_\alpha=-\widehat \Lambda(d h_\alpha)$, with
$\alpha=1,\ldots,5$ and 
\begin{equation}\label{equFun}
\begin{gathered}
h_1(x,p)=-2\sqrt{-p},\quad\quad
h_2(x,p)=p,\quad\quad h_3(x,p)=xp,\quad\quad
h_4(x,p)=x^2p,\\
h_5(x,p)=-2x\sqrt{-p}.\\
\end{gathered}
\end{equation}

We can also show that second-order Kummer--Schwarz equations \cite{Be07,CGL11},
i.e. the equations
$$	
\frac{d^2x}{dt^2}=\frac{3}{2x}\left(\frac{dx}{dt}\right)^2-2c_0 x^3+2b_1(t)x,
$$
with $c_0$ a constant and $b_1(t)$ an arbitrary function of the time, admit
similar descriptions. By using {\it Jacobi multipliers} \cite{CLR11}, it can easily be derived a
$t$-dependent non-natural Lagrangian
$$
L(t,x,v)=\frac{v^2}{x^3}-4c_0x-\frac{4b_1(t)}{x}
$$
for these equations. This Lagrangian induces a Legendre transformation
$$p=\frac{2\,v}{x^3}\Longrightarrow  v=\frac{p\,x^3}2,
$$
for which the induced $t$-dependent Hamiltonian turns out to be 
$$h(t,x,p)=\frac 14 p^2\,x^3+4c_0x+\frac{4b_1(t)}{x}.
$$
Therefore, the Legendre transformation maps the
Kummer--Schwarz equations (written as first-order systems) into the Hamilton
equations
\begin{equation}\label{Hamil2}
\left\{
\begin{aligned}
\frac{dx}{dt}&=\frac{px^3}{2},\\
\frac{dp}{dt}&=-\frac{3p^2x^2}{4}-4c_0+\frac{4b_1(t)}{x^2},
\end{aligned}\right.
\end{equation}
on ${\rm T}^*\mathbb{R}_0$, where $\mathbb{R}_0=\mathbb{R}-\{0\}$. Once again,
the above system is a Lie system as it describes the integral curves of the
$t$-dependent vector field  $X_t=X_3+b_1(t)X_1$, where
\begin{equation}\label{2KSVecFiel}
X_1=\frac{4}{x^2}\frac{\partial}{\partial p},\quad
X_2=x\frac{\partial}{\partial x}-p\frac{\partial}{\partial p},\quad
X_3=\frac{px^3}{2}\frac{\partial}{\partial
x}-\left(\frac{3p^2x^2}{4}+4c_0\right) \frac{\partial}{\partial p},
\end{equation}
span a three-dimensional Lie algebra $V^{2KS}$ isomorphic to
$\mathfrak{sl}(2,\mathbb{R})$. Indeed, 
$$
[X_1,X_3]=2X_2,\qquad [X_1,X_2]=X_1,\qquad [X_2,X_3]=X_3.
$$
Apart from providing a new approach to Kummer--Schwarz equations (see
\cite{CGL11} for a related method), our new description possesses an additional
relevant property: $V^{2KS}$ consists of Hamiltonian vector fields with respect
to the Poisson bivector $\Lambda=\partial/\partial x\wedge\partial/\partial p$
on ${\rm T}^*\mathbb{R}_0$. In fact, $X_\alpha=-\widehat\Lambda(dh_\alpha)$ with
$\alpha=1,2,3$ and
\begin{equation}\label{FunKS}
h_1=\frac 4 x,\qquad h_2=xp,\qquad h_3=\frac 14p^2x^3+4c_0x. 
\end{equation}

We now focus on analysing the Hamilton equations for an $n$-dimensional
Winternitz--Smorodinsky oscillator \cite{WSUF67} of the form
\begin{equation}\label{LieS}
\left\{\begin{aligned}
\frac{dx_i}{dt}&=p_i,\\ 
\frac{dp_i}{dt}&=-\omega^2(t)x_i+\frac{k}{x_i^3},
\end{aligned}\right.\qquad i=1,\ldots,n,
\end{equation}
with $\omega(t)$ an arbitrary $t$-dependent function. These oscillators have
attracted quite much attention in classical and quantum mechanics for their
special properties \cite{GPS06,HBS05,YNHJ11}. In addition, observe that, when $k=0$, Winternitz--
Smorodinsky oscillators reduce to $t$-dependent isotropic harmonic oscillators.

System (\ref{LieS}) describes the integral curves of the $t$-dependent vector
field
$$
X_t=\sum_{i=1}^n\left[p_i\frac{\partial}{\partial
x_i}+\left(-\omega^2(t)x_i+\frac{k}{x_i^3}\right)\frac{\partial}{\partial
p_i}\right]
$$
on ${\rm T}^*\mathbb{R}^{n}_0$. This cotangent bundle admits a natural Poisson
bivector $\Lambda$ related to the restriction to ${\rm T}^*\mathbb{R}^n_0$ of
the canonical symplectic structure on ${\rm T}^*\mathbb{R}^n$. If we consider
the vector fields
\begin{equation}\label{VGSec}
\begin{gathered}
X_1=-\sum_{i=1}^nx_i\frac{\partial}{\partial p_i},\qquad\qquad
X_2=\sum_{i=1}^n\frac{1}{2}\left(p_i\frac{\partial}{\partial
p_i}-x_i\frac{\partial}{\partial x_i}\right),\\
X_3=\sum_{i=1}^n\left(p_i\frac{\partial}{\partial
x_i}+\frac{k}{x_i^3}\frac{\partial}{\partial p_i}\right),
\end{gathered}
\end{equation}
we can write $X_t=X_3+\omega^2(t)X_1$. Additionally, since
\begin{equation}\label{relWS}
[X_1,X_3]=2X_2,\qquad [X_1,X_2]=X_1,\qquad [X_2,X_3]=X_3,
\end{equation}
it follows that (\ref{LieS}) is a Lie system related to a Vessiot--Guldberg Lie
algebra isomorphic to $\mathfrak{sl}(2,\mathbb{R})$. In addition, this Lie
algebra is again made of Hamiltonian vector fields. In fact, it is easy to check
that $X_\alpha=-\widehat\Lambda(dh_\alpha)$, with $\alpha=1,2,3$ and 
\begin{equation}\label{HamWS}
h_1=\frac{1}{2}\sum_{i=1}^{n}{x_i^2},\qquad h_2=-\frac
12\sum_{i=1}^{n}{x_ip_i},\qquad
h_3=\frac{1}{2}\sum_{i=1}^{n}{\left(p_i^2+\frac{k}{x_i^2}\right)}.
\end{equation}

Let us now analyse a final example on a Poisson (but non-symplectic) manifold.
Consider the Euler equations on the dual $\mathfrak{g}^*$ of a Lie algebra
$(\mathfrak{g},[\cdot,\cdot]_\mathfrak{g})$, i.e.
\begin{equation}\label{Euler}
\frac{d\theta}{dt}=-{\rm coad}_{\phi(t)}\theta,\qquad \theta\in\mathfrak{g}^*,
\end{equation}
where $\phi(t)$ is a curve in $\mathfrak{g}$ and ${\rm
coad}_{\phi(t)}\theta=-\theta\circ {\rm  ad}_{\phi(t)}\in \mathfrak{g}^*$, which
appear, for instance, in the study of geometric phases for classical systems
\cite{Bo84,FLV10}. 

Take a basis $\{e_1,\ldots,e_r\}$ for $\mathfrak{g}$ with structure constants
$c_{\alpha\beta\gamma}$, i.e.
$[e_\alpha,e_\beta]=\sum_{\gamma=1}^rc_{\alpha\beta\gamma}e_\gamma$ and
$\alpha,\beta=1,\ldots,r$. It is easy to see that the vector fields
$Y_\alpha(\theta)=-{\rm coad}_{e_\alpha}(\theta)\in {\rm
T}_\theta\mathfrak{g}^*$, with $\alpha=1,\ldots,r$, span a Vessiot--Guldberg Lie
algebra $V^E$ for (\ref{Euler}). Indeed, they generate the Lie algebra of
fundamental vector fields of the coadjoint action of a Lie group $G$ with Lie
algebra $\mathfrak{g}$ \cite{FLV10}. Consequently, if we write
$\phi(t)=\sum_{\alpha=1}^rb_\alpha(t)e_\alpha$, then the Euler equations 
describe the integral curves of the $t$-dependent vector fields of the form
$$
X^\phi_t=\sum_{\alpha=1}^rb_\alpha(t)Y_\alpha,
$$
which take values in the finite-dimensional Lie algebra $V^E$. In other words,
the Euler equations are Lie systems.

To prove that $V^E$ consists of Hamiltonian vector fields, we need to endow
$\mathfrak{g}^*$ with a Poisson structure. This can naturally be done through
the so-called {\it Lie--Poisson bracket} on $\mathfrak{g}^*$ \cite{Ol91}. In fact, since
$d f_\theta ,dg_\theta \in ({\rm T}_\theta\mathfrak{g}^*)^*\simeq \mathfrak{g}$
for every pair $f,g\in C^{\infty}(\mathfrak{g}^*)$, it makes sense to define the
 Lie--Poisson bracket as $\{f,g\}_{\mathfrak{g}^*}(\theta)=\langle [df_\theta,dg_\theta]_\mathfrak{g},\theta \rangle$, where $\langle\cdot,\cdot\rangle$
stands for the pairing between elements of $\mathfrak{g}$ and $\mathfrak{g}^*$. 

Having equipped $\mathfrak{g}^*$ with the Poisson bivector
$\Lambda_{\mathfrak{g}^*}$  corresponding to the Lie--Poisson bracket,  a simple
calculation shows that the vector fields $Y_\alpha$ are Hamiltonian (with
respect to $\widehat\Lambda_{\mathfrak{g}^*}$) with Hamiltonian functions
$h_\alpha(\cdot)=-\langle e_\alpha, \cdot\rangle$.

The properties of the above relevant examples suggest us to define the following
particular type of Lie systems.

\begin{definition}
We say that a system $X$ is a {\it Lie--Hamilton system} if $V^X$ is a
finite-dimensional real Lie algebra of Hamiltonian vector fields with respect to
a certain Poisson structure.
\end{definition}

\begin{note} Obsere that the above definition is equivalent to saying that $X$ is a Lie--Hamilton
system if and only if it admits a Vessiot--Guldberg Lie algebra of Hamiltonian vector fields with
respect to a certain Poisson structure. 
\end{note}

\section{Lie--Hamiltonian structures}

Let us further investigate the properties of the examples provided in the
previous section. Consider again the Euler equation (\ref{Euler}). The
Hamiltonian functions $h_\alpha(\cdot)=-\langle e_\alpha,\cdot \rangle$ of the
vector fields $Y_\alpha$ related to Euler equations satisfy 
\begin{equation*}
 \begin{aligned}
\{h_\alpha,h_\beta\}_{\mathfrak{g}^*}(\theta)&=\langle[(dh_\alpha)_\theta
,(dh_\beta)_\theta]_\mathfrak{g},\theta \rangle=
\langle[e_\alpha,e_\beta]_{\mathfrak{g}},\theta \rangle\\
&=\sum_{
\gamma=1}^rc_{\alpha\beta\gamma}\langle
e_\gamma,\theta\rangle=-\sum_{\gamma=1}^rc_{\alpha\beta\gamma}h_\gamma(\theta).
\end{aligned}
\end{equation*}
That is, they are a basis for a finite-dimensional real Lie algebra
$(\mathfrak{W},\{\cdot,\cdot\}_{\mathfrak{g}^*})$ of functions in
$\mathfrak{g}^*$. Additionally, we can write  
$$
X^\phi_t=\sum_{\alpha=1}^rb_\alpha(t)Y_\alpha=\sum_{\alpha=1}
^rb_\alpha(t)\widehat\Lambda_{\mathfrak{g}^*}(-dh_\alpha)=-\widehat\Lambda_{
\mathfrak{g}^*} \left[d\left(\sum_{\alpha=1}^rb_\alpha(t)h_\alpha\right)\right].
$$
In other words, the $t$-dependent vector field $X$ is determined  through the
Poisson bivector $\Lambda_{\mathfrak{g}^*}$ 
and the curve $h_t=\sum_{\alpha=1}^rb_\alpha(t)h_\alpha$ within a
finite-dimensional real Lie algebra of functions. 

Likewise, the remaining examples of Section 4 enjoy a similar property. For
instance, the Hamiltonian functions (\ref{FunKS}) 
and (\ref{HamWS}) related to the Hamilton equations (\ref{Hamil2}) and
(\ref{LieS}) for second-order Kummer--Schwarz equations
 and Winternitz--Smorodinsky oscillators, correspondingly, satisfy the
commutation relations
\begin{equation}\label{WinSmo}
\{h_1,h_3\}_\Lambda=-2h_2,\qquad \{h_1,h_2\}_\Lambda=-h_1,\qquad
\{h_2,h_3\}_\Lambda=-h_3,
\end{equation}
where $\{\cdot,\cdot\}_\Lambda$ stands for the Poisson structure associated to
the Poisson bivector $\Lambda$ of each example. The $t$-dependent vector
fields governing the dynamics of systems (\ref{Hamil2}) and (\ref{LieS}) can therefore be
written in the form $X_t=-\widehat \Lambda\circ d(h_3+d(t)h_1)$, where $d(t),h_1,h_2,h_3$
are the corresponding functions for each problem, e.g. $d(t)=\omega^2(t)$ 
and $h_1,h_2,h_3$ given by (\ref{HamWS})
for the system (\ref{LieS}). This leads us
to define the new following notions. 

\begin{definition} A {\it Lie--Hamiltonian structure} is a triple
$(N,\Lambda,h)$, where $(N,\Lambda)$ stands for  a Poisson manifold and $h$
represents a $t$-parametrised family of functions $h_t:N\rightarrow \mathbb{R}$
such that ${\rm Lie}(\{h_t\}_{t\in\mathbb{R}},\{\cdot,\cdot\}_\Lambda)$ is a
finite-dimensional real Lie algebra.  
\end{definition}

\begin{definition} A $t$-dependent vector field $X$ is said to admit, or
to possess, a Lie--Hamiltonian structure $(N,\Lambda,h)$ if $X_t=-\widehat
\Lambda\circ d h_t$ for all $t\in\mathbb{R}$.
\end{definition}

\begin{proposition}\label{First} If a system $X$ admits a Lie--Hamiltonian
structure, then $X$ is a Lie--Hamilton system. 
\end{proposition}
\begin{proof} Let $(N,\Lambda,h)$ be a Lie--Hamiltonian structure for $X$. In consequence,
${\rm Lie}(\{h_t\}_{t\in\mathbb{R}})$ is a finite-dimensional
 Lie algebra. Moreover, $\{X_t\}_{t\in\mathbb{R}}\subset \widehat\Lambda\circ d
[{\rm Lie}(\{h_t\}_{t\in\mathbb{R}})]$, and  as 
 $\widehat \Lambda\circ d $ is a Lie algebra morphism, it follows that
$V=\widehat\Lambda\circ d [{\rm Lie}(\{h_t\}_{t\in\mathbb{R}})]$
  is a finite-dimensional Lie algebra of Hamiltonian vector fields containing
$\{X_t\}_{t\in\mathbb{R}}$. Therefore, $V^X\subset V$ and $X$ is
   a Lie--Hamilton system.
 
\end{proof}

Observe that every Lie--Hamiltonian structure $(N,\Lambda,h)$ induces a unique
Lie--Hamilton system  $X_t=-\widehat\Lambda\circ dh_t$
 admitting it as a Lie--Hamiltonian structure. This is interesting as
Lie--Hamiltonian structures appear in the physics literature and they 
 therefore allow us to determine Lie--Hamilton systems of interest
\cite{BBHMR09,BBR06,BCR96,BR98}. For instance,  consider a Lie algebra
 morphism $D:(\mathfrak{g},[\cdot,\cdot]_{\mathfrak{g}})\rightarrow
(C^\infty(T^*M),\{\cdot,\cdot\})$, where  $\mathfrak{g}$ is a 
 finite-dimensional real Lie algebra and $\{\cdot,\cdot\}$ is the canonical
Poisson structure on $T^*M$
 defined by its natural symplectic structure. This is the case when we have a strongly
Hamiltonian  action of $G$ on $T^*M$ (see \cite{LM87}), i.e. a comomentum map
which is additionally a Lie algebra homomorphism. When choosing   a basis
$e_1,\ldots,e_r$ for $\mathfrak{g}$, we can define  a $t$-dependent Hamiltonian
of the form
\begin{equation}\label{RepHam}
h_t=\sum_{\alpha=1}^rb_\alpha(t)D(e_\alpha).
\end{equation}
 As the  Lie algebra ${\rm Lie}(\{h_t\}_{t\in\mathbb{R}},\{\cdot ,\cdot \})$ is included in the
finite-dimensional real Lie algebra $D(\mathfrak{g})$,
  then ${\rm Lie}(\{h_t\}_{t\in\mathbb{R}})$ is finite-dimensional and for every
curve $h_t\subset {\rm Lie}(\{h_t\}_{t\in\mathbb{R}})$
   the triple $(N,\Lambda,h)$  is a Lie--Hamiltonian structure.  Hamiltonians of
the form (\ref{RepHam}) appear in the physics literature 
   \cite{ADR12,BBHMR09,BBR06,BCR96,BR98,Ka98,KK81} and, in view of Proposition
\ref{First}, they give rise to new Lie--Hamilton 
   systems that can be studied through our techniques.

Let us now analyse the relations between a system $V^X$ and the Lie algebra
${\rm Lie}(\{h_t\}_{t\in\mathbb{R}},\{\cdot,\cdot\}_\Lambda)$ for a system $X$ 
admitting a Lie--Hamiltonian structure $(N,\Lambda,h)$.

\begin{lemma}\label{IsoRule} Given a system $X$ on $N$ possessing an 
Lie--Hamiltonian structure $(N,\Lambda,h)$, we have that
\begin{equation}\label{exacseq}
0\hookrightarrow {\rm Cas}(N,\Lambda)\cap {\rm
Lie}(\{h_t\}_{t\in\mathbb{R}})\hookrightarrow {\rm
Lie}(\{h_t\}_{t\in\mathbb{R}})\stackrel{\mathcal{J}_\Lambda}{\longrightarrow}
V^X\rightarrow 0,
\end{equation}
where $\mathcal{J}_\Lambda: f\in {\rm Lie}(\{h_t\}_{t\in\mathbb{R}})\mapsto
\widehat\Lambda\circ df\in V^X$, is an exact sequence of Lie algebras.

\end{lemma}
\begin{proof} Consider the exact sequence of (generally) infinite-dimensiaonl real Lie
algebras
$$
0\hookrightarrow {\rm Cas}(N,\Lambda)\hookrightarrow
C^\infty(N)\stackrel{\widehat\Lambda\circ d}{\longrightarrow}{\rm
Ham}(N,\Lambda)\rightarrow 0.
$$
Since $X_t=-\widehat\Lambda\circ d h_t$, we see that $V^X={\rm Lie}(\widehat
\Lambda \circ d(\{h_t\}_{t\in\mathbb{R}}))$. 
Using that $\widehat \Lambda\circ d$ is a Lie algebra morphism, 
we have $V^X=\widehat \Lambda \circ d[{\rm
Lie}(\{h_t\}_{t\in\mathbb{R}})]=\mathcal{J}_\Lambda({\rm
Lie}(\{h_t\}_{t\in\mathbb{R}}))$. 
Additionally, as $\mathcal{J}_\Lambda$ is the restriction to ${\rm
Lie}(\{h_t\}_{t\in\mathbb{R}})$ of $\widehat\Lambda\circ d$, 
we obtain that its kernel consists of Casimir functions belonging  to
 ${\rm Lie}(\{h_t\}_{t\in\mathbb{R}})$, i.e. $\ker \mathcal{J}_\Lambda={\rm
Lie}(\{h_t\}_{t\in\mathbb{R}})\cap {\rm Cas}(N,\Lambda)$. 
 The exactness of sequence (\ref{exacseq})  easily follows from these results. 
\end{proof}

The above proposition entails that every system $X$ that possesses a Lie--Hamiltonian
structure $(N,\Lambda,h)$ is such 
that ${\rm Lie}(\{h_t\}_{t\in\mathbb{R}})$ is a Lie algebra extension of $V^X$
by ${\rm Cas}(N,\Lambda)\cap {\rm Lie}(\{h_t\}_{t\in\mathbb{R}})$, i.e. the
sequence of Lie algebras (\ref{exacseq}) is exact. Note that if $X$ is a Lie system, all the Lie
algebras appearing in such a sequence are finite-dimensional.
 For instance, the first-order system (\ref{FORiccati}) associated to
second-order Riccati equations admits a Lie--Hamiltonian structure 
$$\left(\mathcal{O},\frac{\partial}{\partial x}\wedge \frac{\partial}{\partial
p},h_1-a_0(t)h_2-a_1(t)h_3-a_2(t)h_4\right),$$ where ${\rm
Lie}(\{h_t\}_{t\in\mathbb{R}})$, for generic functions $a_0(t),a_1(t),a_2(t)$,
is a six-dimensional Lie algebra of functions $\mathfrak{W}\simeq V^X
\oplus\mathbb{R}$. 

It is worth noting that every $t$-dependent vector field that admits a
Lie--Hamiltonian structure necessarily possesses many other Lie--Hamiltonian
structures. For instance, if system $X$ admits $(N,\Lambda,h)$, then it also admits a Lie--Hamiltonian structure
$(N,\Lambda,h')$, with $h':(t,x)\in \mathbb{R}\times N\mapsto
h(t,x)+f_\mathcal{C}(x)\in \mathbb{R}$, where $f_\mathcal{C}$ is any Casimir
function with respect to $\Lambda$. Indeed, it is easy to see that if
$h_1,\ldots,h_r$ is a basis for ${\rm Lie}(\{h_t\}_{t\in\mathbb{R}})$, then
$h_1,\ldots,h_r,f_\mathcal{C}$ span ${\rm Lie}(\{h'_t\}_{t\in\mathbb{R}})$,
which also becomes a finite-dimensial real Lie algebra. As shown later, this has relevant implications
for the linearisation of Lie--Hamilton systems.

We have already proved that every system $X$ admitting a Lie--Hamiltonian
structure must possess several ones. Nevertheless, we have not yet studied 
the conditions ensuring that a Lie--Hamilton system $X$ possesses a Lie--Hamiltonian structure. Let
us answer this question.

\begin{proposition} Every Lie--Hamilton system admits a Lie--Hamiltonian
structure. 
\end{proposition} 
\begin{proof}
Assume $X$ to be a Lie--Hamilton system on a manifold $N$ with respect to a
Poisson bivector $\Lambda$. 
Since $V^X\subset {\rm Ham}(N,\Lambda)$ is finite-dimensional, there exists a
finite-dimensional linear space $\mathfrak{W}_0\subset C^\infty(N)$
 isomorphic to $V^X$ and such that $\widehat\Lambda\circ d(\mathfrak{W}_0)=V^X$.
Consequently, there exists a curve $h_t$ in $\mathfrak{W}_0$ 
 such that
$X_t=-\widehat{\Lambda}\circ d (h_t)$. To ensure that $h_t$ gives rise to a
Lie--Hamiltonian structure, we need to demonstrate that 
${\rm Lie}(\{h_t\}_{t\in\mathbb{R}},\{\cdot,\cdot\}_\Lambda)$ is
finite-dimensional. This will be done by constructing a finite-dimensional Lie algebra of functions containing the curve $h_t$.

Define the linear isomorphism $T:X_f\in V^X\mapsto -f\in \mathfrak{W}_0\subset
C^\infty(N)$ associating each vector field in $V^X$ 
with minus its unique Hamiltonian function within $\mathfrak{W}_0$.  This can be
done by choosing a representative
 for each element of a basis of $V^X$ and extending the map by linearity.

Note that this mapping needs not be a Lie algebra morphism and
 hence ${\rm Im}\, T=\mathfrak{W}_0$ does not need to be a Lie algebra. Indeed, we can
define a bilinear map 
 $\Upsilon: V^X\times V^X\rightarrow C^\infty(N)$ of the form
\begin{equation}\label{formula2}
\Upsilon(X_f, X_g)=\{f,g\}_\Lambda -T[X_f,X_g],
\end{equation}
measuring the obstruction for  $T$ to be a Lie algebra morphism,
 i.e. $\Upsilon$ is identically null if and only if $T$ is a Lie algebra
morphism. In fact, 
 if $\mathfrak{W}_0$ were a Lie algebra, then $\{f,g\}_\Lambda$ would be the
only element of $\mathfrak{W}_0$ with Hamiltonian 
 vector field $-[X_f,X_g]$, i.e. $T[X_f,X_g]$, and $\Upsilon$ would be a zero function. 

Note that $\Upsilon(X_f,X_g)$ is the difference between two functions, namely
$\{f,g\}_\Lambda$ and $T[X_f,X_g]$, sharing the same Hamiltonian vector field.
Consequently, ${\rm Im}\,\Upsilon\subset {\rm Cas}(N,\Lambda)$ and it can be injected
into a finite-dimensional Lie algebra of Casimir functions of the form
$$
\mathfrak{W}_\mathcal{C}\equiv \langle \Upsilon(X_i,X_j)\rangle,\qquad
i,j=1,\ldots,r,
$$
where $X_1,\ldots,X_r$ is a basis for $V^X$. From here, it follows that  
$$
\{\mathfrak{W}_\mathcal{C},\mathfrak{W}_\mathcal{C}\}_\Lambda=0,\quad
\{\mathfrak{W}_\mathcal{C},\mathfrak{W}_0\}_\Lambda=0,\quad
\{\mathfrak{W}_0,\mathfrak{W}_0\}_\Lambda\subset
\mathfrak{W}_\mathcal{C}+\mathfrak{W}_0.
$$
Hence, $\mathfrak{W}\equiv \mathfrak{W}_0+\mathfrak{W}_\mathcal{C}$ is a
finite-dimensional Lie algebra of functions containing the curve $h_t$. 
From here, it readily follows that $X$ admits a Lie--Hamiltonian structure
$(N,\Lambda,-TX_t)$. 
\end{proof}

Since every Lie--Hamilton system possesses a Lie--Hamiltonian
structure and every Lie--Hamiltonian structure determine a Lie--Hamilton
systems, we obtain  the following theorem.

\begin{theorem}\label{HamLieSys}$\!$ A system $X\!$  admits a Lie--Hamiltonian
structure if and only if it is a Lie--Hamilton system.
\end{theorem}

\section{On general properties of Lie--Hamilton systems}

We now turn to describing the analogue for Lie--Hamilton systems of the basic
properties of general Lie systems. Additionally, we show how the Poisson
structures associated to Lie--Hamilton systems allow us to investigate their
$t$-independent constants of motion, Lie symmetries, superposition rules and
linearisation properties. 

Recall that, as for every Lie system, the general solution $x(t)$ of a
Lie--Hamilton system $X$ on $N$ can be brought into the form
$x(t)=\varphi(g(t),x_0)$, where $x_0\in N$, the map $\varphi:G\times
N\rightarrow N$ is the action of a connected Lie group $G$ whose space of fundamental vector fields
is $V^X\simeq T_eG$, and $g(t)$ is the solution of a Lie system of the form
(\ref{EquLie}). In addition, for a Lie--Hamilton system, the infinitesimal action 
associated to $\varphi$, let us say $\rho^X:\mathfrak{g}\rightarrow \Gamma(\tau_N)$, takes also
values in a certain space ${\rm Ham}(N,\Lambda)$. In other words, $\varphi$ is a {\it Hamiltonian} Lie group action. Furthermore, the mappings
$\varphi_g:x\in N\mapsto \varphi(g,x)\in N$, with $g\in G$, are {\it Poisson
maps}, i.e. 
$$
\varphi_{g*}\Lambda=\Lambda.
$$

The above Lie group action plays another relevant r\^ole. It is known that if
$G$ is connected, every curve $\bar g(t)$ in $G$ induces a
 $t$-dependent change of variables mapping a Lie system $X$ taking values in a
Lie algebra $V^X$ into another Lie system $\bar X$, 
 with general solution $\bar x(t)=\varphi(\bar g(t),x(t))$, taking values in the
same Lie algebra $V^X$ \cite{CGL09,CarRamGra,CLLEmd}. 
 In the particular case of $X$ being a Lie--Hamilton system, the vector fields
$\{\bar X_t\}_{t\in\mathbb{R}}$ are also Hamiltonian and $\bar X$ is 
 again a Lie--Hamilton system. 

Using again that $x(t)=\varphi(g(t),x_0)$, we see that the solutions of a
Lie system $X$ are contained in the orbits of $\varphi$. Indeed, it is easy to
see that the vector fields $\{X_t\}_{t\in\mathbb{R}}$ are tangent such orbits.
Therefore, the integration of a Lie system $X$ reduces to integrating its
restrictions to each orbit of $\varphi$, which are Lie systems also. 

Meanwhile, for Lie--Hamilton systems, we have another related method of
reduction. Note that given a Lie--Hamilton system $X$ admitting a Lie--Hamilton
structure $(N,\Lambda,h)$, we have that $\mathcal{D}^X\subset
\mathcal{F}^\Lambda$, where we recall that $\mathcal{F}^\Lambda$ is the
characteristic distribution related to $\Lambda$ and $\mathcal{D}^X$ is spanned by Hamiltonian vector fields within $V^X$. Hence, the vector
fields $\{X_t\}_{t\in\mathbb{R}}$ are tangent to the symplectic leaves of the
Poisson manifold $(N,\Lambda)$. From here, it immediately follows the theorem
below.

\begin{theorem} The integration of a Lie--Hamilton system $X$ possessing a Lie--Hamiltonian
structure $(N,\Lambda,h)$ reduces to integrating the restrictions of
$X|_{\mathfrak{F}^\Lambda}$ to every symplectic leaf $\mathfrak{F}^\Lambda$
associated to $(N,\Lambda)$. Every such a system is a Lie--Hamilton system with
respect to the symplectic structure
$(\mathfrak{F}^\Lambda,\omega_{\mathfrak{F}^\Lambda})$ induced by the
restriction of $\Lambda$ to $\mathfrak{F}^\Lambda$.
\end{theorem}
 
Let us now turn to describing several properties of constants of motion for Lie
systems.

\begin{proposition}\label{Cas}
Given a system $X$ with a Lie--Hamiltonian structure $(N,\Lambda,h)$,
then $\mathcal{C}^\Lambda\subset \mathcal{V}^X$, where we recall that $\mathcal{C}^\Lambda$ is the Casimir distribution relative to $\Lambda$.  
\end{proposition}

\begin{proof}
Consider a $\theta_x\in\mathcal{C}_x^\Lambda$, with $x\in N$. As $X$ is a Lie--Hamilton system, 
for every $Y\in V^X$ there exists a function $f\in C^\infty(N)$ such that
$Y=-\widehat\Lambda (df)$. Then,
$$
\theta_x(Y_x)=-\theta_x(\widehat{\Lambda}_x(df_x))=-\Lambda_x(df_x,\theta_x)=0,
$$
where $\widehat{\Lambda}_x$ is the restriction of $\widehat \Lambda$ to
$T^*_xN$. As the vectors $Y_x$, with $Y\in V^X$, span $\mathcal{D}^X_x$, then
$\theta_x\in \mathcal{V}_x^X$ and $\mathcal{C}^\Lambda\subset \mathcal{V}^X$. 
\end{proof}

Observe  that different Lie--Hamiltonian structures for a Lie--Hamilton system
$X$ may lead to different families of Casimir functions, which may determine
different constants of motion for $X$.

\begin{theorem}\label{IntLie} Let $X$\! be a system admitting a Lie--Hamiltonian
structure $(N,\!\Lambda,h)$, the space $\mathcal{I}^X|_U\!$ of $t$-independent
constants of motion of $X\!$\! on an open $U\!\subset\! U^X$ is a Poisson algebra.
Additionally, the codistribution $\mathcal{V}^X|_{U_X}$ is involutive with respect to
the Lie bracket $[\cdot,\cdot ]_\Lambda$ induced by $\Lambda$ on
$\Gamma(\pi_N)$. 
\end{theorem}
\begin{proof}
Let  $f_1,f_2:U\rightarrow\mathbb{R}$ be two $t$-independent functions constants of motion
for $X$, i.e. $X_tf_i=0$, for $i=1,2$ and $t\in \mathbb{R}$. As $X$ is a
Lie--Hamilton system, all the elements of $V^X$ are Hamiltonian vector fields
and we can write $Y\{f,g\}_\Lambda=\{Yf,g\}_\Lambda+\{f,Yg\}_\Lambda$ for every
$f,g\in C^\infty(N)$. In particular,
$X_t(\{f_1,f_2\}_\Lambda)=\{X_tf_1,f_2\}_\Lambda+\{f_1,X_tf_2\}_\Lambda=0$, i.e. the
Poisson bracket of $t$-independent constants of motion is a new one.
As $\lambda f_1+\mu f_2$ and $f_1\cdot f_2$ are also $t$-independent constants
of motion for every $\lambda,\mu\in\mathbb{R}$, it easily follows that
$\mathcal{I}^X|_U$ is a Poisson algebra. 

In view of 
Lemma \ref{basisVX}, the co-distribution $\mathcal{V}^X$ admits a local basis of
exact forms $df_1,\ldots,df_{p(x)}$ for every point $x\in U^X$, where
$\mathcal{V}^X$ has local constant rank $p(x)=\dim\, N-\dim\,\mathcal{D}^X_x$.  Now, 
$[df_i,df_j]_\Lambda=d(\{f_i,f_j\}_\Lambda)$ for $i,j=1,\ldots,p(x)$. We already proved that the
function $\{f_i,f_j\}_\Lambda$ is another first-integral. Therefore, in view of Lemma
\ref{basisVX}, it easily follows that $\{f_i,f_j\}_\Lambda=G(f_1,\ldots,f_{p(x)})$.
Thus, $[df_i,df_j]_\Lambda \in\mathcal{V}^X|_{U_X}$. From here and using the properties
of the Lie bracket $[\cdot,\cdot]_\Lambda$, it directly turns out that the Lie bracket of two one-forms
taking values in $\mathcal{V}^X|_{U_X}$ belongs to $\mathcal{V}^X|_{U_X}$. Hence, 
$\mathcal{V}^X|_{U_X}$ is involutive. 
\end{proof}
\begin{corollary} Given a Lie--Hamilton system $X$, the space
$\mathcal{I}^X|_U$, where $U\subset U^X$ is such that $\mathcal{V}^X$ admits a local basis of exact forms, is a function
group, that is: 
\begin{enumerate}
 \item The space $\mathcal{I}^X|_U$ is a Poisson algebra.
\item There exists a family of functions $f_1,\ldots,f_s\in\mathcal{I}^X|_U$ such that 
every element $f$ of $\mathcal{I}^X|_U$ can be put in the form
$f=F(f_1,\ldots,f_s)$ for a certain function
$F:\mathbb{R}^s\rightarrow\mathbb{R}$.
\end{enumerate}
\end{corollary}
\begin{proof}
In view of the previous theorem, $\mathcal{I}^X|_U$ is a Poisson algebra with
respect to a certain Poisson bracket. Taking into account Proposition \ref{NuX} and the form of
$\mathcal{I}^X|_U$ given by Lemma \ref{basisVX}, we obtain that this space becomes
a function group.
\end{proof}

The above properties do not necessarily hold for systems other than
Lie--Hamilton systems, as they do not need to admit any {\it a priori} relation
among a Poisson bracket of functions and the $t$--dependent vector field
describing the system. Let us exemplify this. Consider the Poisson manifold
$(\mathbb{R}^3,\Lambda_{GM})$, where 
$$
\Lambda_{GM}=\sigma_3\frac{\partial}{\partial
\sigma_2}\wedge\frac{\partial}{\partial
\sigma_1}-\sigma_1\frac{\partial}{\partial
\sigma_2}\wedge\frac{\partial}{\partial
\sigma_3}+\sigma_2\frac{\partial}{\partial
\sigma_3}\wedge\frac{\partial}{\partial \sigma_1}
$$
and $(\sigma_1,\sigma_2,\sigma_3)$ is a coordinate basis for $\mathbb{R}^3$, 
appearing in the study of Classical XYZ Gaudin Magnets \cite{BR98}. The system
$X=\partial/\partial \sigma_3$ is not a Lie--Hamilton system with respect to
this Poisson structure as $X$ is not Hamiltonian, namely
$\mathcal{L}_X\Lambda_{GM}\neq 0$. In addition, this system admits two
first-integrals $\sigma_1$ and $\sigma_2$. Nevertheless, their Lie bracket reads
$\{\sigma_1,\sigma_2\}=-\sigma_3$, which is not a first-integral for $X$. On the
other hand, consider the system
$$
Y=\sigma_3\frac{\partial}{\partial \sigma_2}+\sigma_2\frac{\partial}{\partial
\sigma_3}.
$$
This system is a Lie--Hamilton system, as it can be written in the form
$Y=-\widehat\Lambda_{GM}(d\sigma_1)$, and it possesses two first-integrals given
by $\sigma_1$ and $\sigma_2^2-\sigma_3^2$. Unsurprisingly,
$Y\{\sigma_1,\sigma^2_2-\sigma_3^2\}=0$, i.e. the Lie bracket of two
$t$-independent constants of motion is also a constant of motion. 

Let us prove some final interesting results about the $t$-independent constants
of
motion for Lie--Hamilton systems.
\begin{proposition} Let $X$ be a Lie--Hamilton system that admits a
Lie--Hamiltonian structure $(N,\Lambda,h)$. The function $f:N\rightarrow
\mathbb{R}$ is a constant of motion for $X$ if and only if $f$ Poisson commutes
with all elements of ${\rm
Lie}(\{h_t\}_{t\in\mathbb{R}},\{\cdot,\cdot\}_{\Lambda})$.
\end{proposition}
\begin{proof}
The function $f$ is a $t$-independent constant of motion for $X$ if and only if 
\begin{equation}\label{con2}
0=X_tf=\{f,h_t\}_{\Lambda},\qquad \forall t\in\mathbb{R}.
\end{equation}
From here,
$$
\{f,\{h_t,h_{t'}\}_{\Lambda}\}_{\Lambda}=\{\{f,h_t\}_{\Lambda},h_{t'}\}_{\Lambda
}+\{h_t,\{f,h_{t'}\}_{\Lambda}\}_{\Lambda}=0,\qquad \forall t,t'\in\mathbb{R},
$$
and inductively follows that $f$ Poisson commutes with all  successive Poisson
brackets of elements of $\{h_t\}_{t\in\mathbb{R}}$ and their linear combinations. 
As these elements span ${\rm
Lie}(\{h_t\}_{t\in\mathbb{R}})$, we get that $f$ Poisson commutes with ${\rm
Lie}(\{h_t\}_{t\in\mathbb{R}})$. 

Conversely, if $f$ Poisson commutes with ${\rm Lie}(\{h_t\}_{t\in\mathbb{R}})$,
it Poisson commutes with the elements $\{h_t\}_{t\in\mathbb{R}}$, and, in view
of (\ref{con2}), it becomes a constant of motion for $X$.
\end{proof}

In order to illustrate the above proposition, let us consider a Winternitz--Smorodinsky system 
(\ref{LieS}) with $n=2$. Recall that this system admits a Lie--Hamiltonian structure $({\rm T}^*\mathbb{R}_0^2,\Lambda,h=h_3+\omega^2(t)h_1)$, 
where $\Lambda=\sum_{i=1}^2\partial/\partial x_i\wedge\partial/\partial p_i$ is a  Poisson bivector 
on ${\rm T}^*\mathbb{R}_0^2$ and the functions $h_1,h_3$ are given within (\ref{HamWS}). 
For non-constant $\omega(t)$, it is easy to prove that 
${\rm Lie}(\{h_t\}_{t\in\mathbb{R}},\{\cdot,\cdot\}_\Lambda)$ is a real Lie algebra of functions isomorphic
to $\mathfrak{sl}(2,\mathbb{R})$ generated by the functions $h_1,h_2$ and $h_3$ detailed in (\ref{HamWS}).
When $\omega(t)=\omega_0\in\mathbb{R}$, the Lie algebra ${\rm Lie}(\{h_t\}_{t\in\mathbb{R}},\{\cdot,\cdot\}_\Lambda)$ becomes a one-dimensional
Lie subalgebra of the previous one. In any case, it is known that 
\begin{equation}\label{Inte}
I=(x_1p_2-p_1x_2)^2+k\left[\left(\frac{x_1}{x_2}\right)^2+\left(\frac{x_2}{x_1}\right)^2\right]
\end{equation}
is a $t$-independent constant of motion (cf. \cite{CLR08}). A simple calculation shows that
$$
\{I,h_\alpha\}_\Lambda=0,\qquad \alpha=1,2,3.
$$ 
Then, the function $I$ always Poisson commutes with the whole
Lie algebra ${\rm Lie}(\{h_t\}_{t\in\mathbb{R}},\{\cdot,\cdot\}_\Lambda)$, as expected. 

Obviously, every autonomous Hamiltonian system is a Lie--Hamilton system possessing a Lie--Hamiltonian structure
$(N,\Lambda,h)$ with $h$ being a time-independent Hamiltonian. Consequently, above proposition shows that the time-independent 
first-integrals for a Hamiltonian system are those functions that Poisson commute with its Hamiltonian, recovering as a particular
case this wide-known result. 

Moreover, above proposition suggests us that the r\^ole played by autonomous Hamiltonians for Hamiltonian systems is performed by the finite-dimensional Lie algebras of functions associated with Lie--Hamiltonian structures in the case of Lie--Hamilton systems. This can be employed, for instance, to study time-independent first-integrals of Lie--Hamilton systems or, more specifically, the maximal number of such first-integrals in involution, which would lead to the interesting analysis of integrability/superintegrability of Lie--Hamilton systems.

\begin{definition}
We say that a Lie system $X$ admitting a Lie--Hamilton structure $(N,\Lambda,h)$ possesses a compatible
 {\it strong comomentum map} with respect to this Lie--Hamilton structure if there
exists a Lie algebra morphism
$\lambda:V^X\rightarrow {\rm Lie}(\{h_t\}_{t\in\mathbb{R}},\{\cdot,\cdot \}_\Lambda)$ such that the following diagram:
$$
\xymatrix{&&V^X\ar[d]^\iota\ar[lld]_\lambda\\
{\rm Lie}(\{h_t\}_{t\in\mathbb{R}},\{\cdot,\cdot \}_\Lambda)\ar[rr]^{\widehat\Lambda\circ d}&&{\rm Ham}(N,\Lambda)
}
$$
where $\iota:V^X\hookrightarrow {\rm Ham}(N,\Lambda)$ is the natural injection
of $V^X$ into ${\rm Ham}(N,\Lambda)$, is commutative.
\end{definition}

Observe that, in this case, $X$ induces a Hamiltonian Lie group action $\varphi:G\times N\rightarrow N$ whose set of fundamental vector 
fields  lies in $V^X$ and admits a comomentum map $\lambda:V^X\rightarrow {\rm Lie}(\{h_t\}_{t\in\mathbb{R}},\{\cdot,\cdot \}_\Lambda)\subset C^\infty(N)$ that is a Lie algebra morphism, i.e.  we say that  $\varphi$ is a {\it strongly Hamiltonian action}. 
Note additionally that given $Y\in V^X$, we have $\widehat \Lambda\circ d\circ \lambda(Y)=Y$, i.e. $-\lambda(Y)$ is a Hamiltonian function for $Y$.

Conversely, it can easily be proved that if $X$ is a Lie system inducing such 
a strongly Hamiltonian Lie group action, then $X$ is a Lie--Hamilton system admitting a Lie--Hamiltonian structure $(N,\Lambda,-\lambda(X_t))$ that is compatible with a strong comomentum map $\lambda$.

It is important to note that if $X$ possesses a Lie--Hamiltonian structure $(N,\Lambda,h)$ compatible with a strong comomentum
 map $\lambda$, then the Lie algebra $\lambda(V^X)$ is isomorphic to $ V^X$  and it can readily be proved that  $X$ admits an additional 
Lie--Hamiltonian structure $(N,\Lambda,\bar h_t\equiv-\lambda (X_t))$ satisfying that $V^X\simeq {\rm Lie}(\{\bar h_t\}_{t\in\mathbb{R}})$ and admitting $\lambda$ as a compatible strong comomentum map. 

Let us provide a particular example of a strong comomentum map for 
a second-order Kummer--Schwarz equation in Hamiltonian form (\ref{Hamil2}) 
with a non-constant $\omega(t)$. 
Since the corresponding $t$-dependent vector field $X$ satisfies that $X_t=X_3+\omega^2(t)X_1$, where $X_1,X_3$ are given by (\ref{2KSVecFiel}), 
and in view of the relations (\ref{relWS}),
the Lie algebra $V^X$ is 
 isomorphic to $\mathfrak{sl}(2,\mathbb{R})$. Recall that this system admits a
Lie--Hamiltonian structure $({\rm T}^*\mathbb{R}_0,\Lambda,h_t=h_3+\omega^2(t)h_1)$, with $h_1$ and $h_3$ given in
(\ref{2KSVecFiel}). It is easy to see that
$X$ admits a strong 
comomentum map, relative to the previous Lie--Hamiltonian structure,  
$\lambda:V^X\rightarrow C^\infty({\rm T}^*\mathbb{R}_0)$ such that 
$\lambda(X_\alpha)=-h_\alpha$, where $X_1,X_2, X_3$ and $h_1,h_2,h_3$ are given by 
(\ref{2KSVecFiel}) and (\ref{FunKS}), correspondingly. As expected, $V^X$ and 
$\lambda(V^X)={\rm Lie}(\{h_t\}_{t\in\mathbb{R}},\{\cdot,\cdot\}_\Lambda)$
are isomorphic Lie algebras. A similar result can readily be obtained for $\omega(t)=\omega_0$, with $\omega_0\in\mathbb{R}$.
In this case, we now have that $V^X$ and the related 
${\rm Lie}(\{h_t\}_{t\in\mathbb{R}},\{\cdot,\cdot\}_\Lambda)$ become one-dimensional Lie algebras.

\begin{proposition}\label{linearization} Let $X$ be a Lie system possessing a Lie--Hamiltonian structure $(N,\Lambda,h)$
compatible with a strong comomentum map $\lambda$  such that $\dim\,\mathcal{D}^X_x=\dim
N=\dim V^X$ at a generic $x\in N$. Then, there exists a local coordinate system defined on a neighbourhood of each $x$
such that $X$ and $\Lambda$ are
simultaneously linearisable and where  $X$ possesses a linear superposition rule. 
\end{proposition}
\begin{proof}
As it is assumed that $n\equiv\dim N=\dim V^X=\dim \mathcal{D}^X_x$ at a generic $x$, every basis $X_{1},\ldots,X_{n}$
of $V^X$ gives rise to a basis for the tangent bundle $TN$ on a neighbourhood of $x$. 
Since $X$ admits a strong comomentum map compatible with $(N,\Lambda,h)$, 
we have $(V^X,[\cdot,\cdot])\simeq (\lambda(V^X),\{\cdot,\cdot\}_\Lambda)$ and the family of 
functions, $h_\alpha=\lambda(X_\alpha)$, with $\alpha=1,\ldots,n,$ form a basis for 
the Lie subalgebra $\lambda(V^X)$. Moreover, since $\widehat\Lambda\circ d\circ \lambda (V^X)=V^X$ and 
$\dim\,V^X=\dim \mathcal{D}^X_{x'}$  for $x'$ 
 in a neighbourhood of $x$,
 then $\widehat \Lambda_{x'}\circ d (\lambda(V^X))\simeq T_{x'}N$
 and $dh_1\wedge \ldots \wedge dh_n\neq 0$ at a generic point. Hence, the set $(h_1,\ldots,h_n)$ is a
coordinate system on an open dense subset of $N$. Now, using again that $(\lambda(V^X), \{\cdot,\cdot\}_\Lambda)$ 
is a real Lie algebra, the Poisson bivector
$\Lambda$ can be put in the form
\begin{equation}\label{linear}
\Lambda=\frac 12\sum_{i,j=1}^n\{h_i,h_j\}_\Lambda\frac{\partial}{\partial
h_i}\wedge\frac{\partial}{\partial h_j}=\frac
12\sum_{i,j,k=1}^nc_{ijk}h_k\frac{\partial}{\partial
h_i}\wedge\frac{\partial}{\partial h_j},
\end{equation}
for certain real $n^3$ constants $c_{ijk}$. In other words, the Poisson bivector
$\Lambda$ becomes linear in the chosen coordinate system.

Since we can write $X_t=-\widehat\Lambda (d\bar h_t)$, with $\bar{h}_t=-\lambda (X_t)$ being a curve in the Lie
 algebra $\lambda(V^X)\subset {\rm
Lie}(\{h_t\}_{t\in\mathbb{R}})$, expression (\ref{linear}) yields
$$
\begin{aligned}
X_t&=-\widehat\Lambda(d\bar h_t)=-\widehat\Lambda\circ
d\left(\sum_{l=1}^nb_l(t)h_l\right)\\&=-\sum_{l=1}^nb_l(t)(\widehat\Lambda\circ
dh_l)=-\sum_{l,j,k=1}^nb_l(t)c_{ljk}h_k\frac{\partial}{\partial h_j},
\end{aligned}
$$
and $X_t$  is  linear in this  coordinate system. Consequently, as every linear system, $X$ admits a linear superposition rule in the coordinate system
$(h_1,\ldots,h_n)$.
\end{proof}

Let us turn to describing some features of $t$-independent Lie symmetries for
Lie--Hamilton systems. Our exposition will be based upon the properties of the
hereafter called {\it symmetry distribution}.

\begin{definition} Given a Lie--Hamilton system $X$ that possesses a Lie--Hamiltonian structure $(N,\Lambda,h)$, we
define its {\it symmetry distribution}, $\mathcal{S}^X_\Lambda$, by
$$
(\mathcal{S}^X_\Lambda)_x=\widehat{\Lambda}_x(\mathcal{V}_x^X)\in T_xN,\qquad
x\in N.
$$
\end{definition}

As its name indicates, the symmetry distribution can be employed to investigate the $t$-independent Lie symmetries of a Lie--Hamilton system.
Let us give some basic examples of how this can be done.

\begin{proposition} Given a Lie--Hamilton system $X$ with a Lie--Hamiltonian
structure $(N,\Lambda,h)$, then:
\begin{enumerate}
 \item The symmetry distribution $\mathcal{S}^X_\Lambda$ associated with $X$ and $\Lambda$ is involutive on $U^X$, i.e. the open dense subset of $N$ where $\mathcal{V}^X$ is differentiable.
\item  If $f$ is a $t$-independent constant of motion for $X$, then $\widehat
\Lambda(df)$ is a $t$-independent Lie symmetry of $X$.
\item The distribution $\mathcal{S}^X_\Lambda$ admits a local basis of $t$-independent Lie
symmetries of $X$ defined around a generic point of $N$. The elements of such a basis
are Hamiltonian vector fields of $t$-independent constants of motion of $X$.
\end{enumerate}
\end{proposition}
\begin{proof}
By definition of $\mathcal{S}^X_\Lambda$ and using that $\mathcal{V}^X$ has constant rank on the 
connected components of $U_X$, we can ensure that 
given two vector fields in $Y_1,Y_2\in \mathcal{S}^X_\Lambda|_{U_X}$, there exist two forms
$\omega,\omega'\in\mathcal{V}^X|_{U_X}$ such that $Y_1=\widehat\Lambda(\omega)$,
$Y_2=\widehat\Lambda(\omega')$. Since $X$ is a Lie--Hamilton system, 
$\mathcal{V}^X|_{U_X}$ is involutive and $\widehat \Lambda$ is an anchor, i.e. a Lie algebra 
morphism from $(\Gamma(\pi_N),[\cdot,\cdot]_\Lambda)$ to $(\Gamma(\tau_N),[\cdot,\cdot])$, then 
$$
[Y_1,Y_2]=[\widehat\Lambda (w),\widehat\Lambda (w')]=\widehat\Lambda
([w,w']_\Lambda)\in\mathcal{S}^X_\Lambda.
$$ 
In other words, since $\mathcal{V}^X$ is involutive on $U_X$, then
$\mathcal{S}^X_\Lambda$ is so, which proves $(1)$.

To prove $(2)$, note that
$$
[X_t,\widehat \Lambda(df)]=-[\widehat\Lambda (dh_t),\widehat
\Lambda(df)]=-\widehat\Lambda (d\{h_t,f\}_\Lambda)=\widehat\Lambda [d(X_tf)]=0.
$$
Finally, the proof of $(3)$ is based upon the fact that $\mathcal{V}^X$ admits,
around a point $x\in U^X\subset N$, a local basis of one-forms 
$df_1,\ldots, df_{p(x)}$, with $f_1,\ldots,f_{p(x)}$ being a family of
$t$-independent constants of motion for $X$ and
$p(x)=\dim\,N-\dim\mathcal{D}_x^X$. From $(2)$, the vector fields $X_{f_1},\ldots, X_{f_{p(x)}}$
form a family of Lie symmetries of $X$ locally spanning $\mathcal{S}^X_\Lambda$.
Hence, we can easily choose among them a local basis for
$\mathcal{S}^X_\Lambda$. 
\end{proof}

As a particular example of the usefulness of the above result, 
let us turn to a two-dimensional
Winternitz--Smorodinsky oscillator $X$ given by (\ref{LieS}) 
and its known constant of motion (\ref{Inte}). In view of the previous proposition, 
$Y=\widehat\Lambda(dI)$ must be
a Lie symmetry for these systems. A little calculation leads to	
\begin{multline*}
Y=2(x_1p_2-p_1x_2)\left(x_2\frac{\partial}{\partial x_1}-x_1\frac{\partial}{\partial x_2}\right)+\\
2\left[(x_1p_2-p_1x_2)p_2+k\frac{x_1^4-x_2^4}{x_1^3x_2^2}\right]\frac{\partial}{\partial p_1}
-2\left[(x_1p_2-p_1x_2)p_1+k\frac{x_1^4-x_2^4}{x_2^3x_1^2}\right]\frac{\partial}{\partial p_2},
\end{multline*}
and it is straightforward to verify that $Y$ commutes with $X_1,X_2,X_3$, given by (\ref{VGSec}), and therefore with every $X_t$, with $t\in\mathbb{R}$, i.e. $Y$ is a Lie symmetry for $X$.
\begin{proposition} Let $X$ be a Lie--Hamilton system possessing a Lie--Hamiltonian
structure $(N,\Lambda,h)$. If $[V^X,V^X]=V^X$ and $Y\in {\rm Ham}(N,\Lambda)$ is
a Lie symmetry of $X$, then $Y\in \mathcal{S}_\Lambda^X$.
\end{proposition}
\begin{proof}
As $Y$ is a $t$-independent Lie symmetry, then  $[Y,X_t]=0$ for every $t\in\mathbb{R}$. Since $Y$ is a
Hamiltonian vector field, then $Y=-\widehat\Lambda\circ df$ for a certain $f\in
C^\infty(N)$. Using that $X_t=-\widehat \Lambda(dh_t)$, we obtain
$$
0=[Y,X_t]=[\widehat\Lambda (df),\widehat \Lambda (dh_t)]=\widehat \Lambda
(d\{f,h_t\}_\Lambda)=\widehat \Lambda [d (X_tf)].
$$
Hence, $X_tf$ is a Casimir function. Therefore, as every $X_{t'}$ is a Hamiltonian vector field for all $t'\in\mathbb{R}$, it turns out that $X_{t'}X_tf=0$ for every
$t,t'\in \mathbb{R}$ and, in consequence, $Z_1f$ is a Casimir function for every
$Z_1\in V^X$. Moreover,  as every $Z_2\in V^X$ is Hamiltonian, we have
$$
Z_2Z_1f=Z_1Z_2f=0\Longrightarrow (Z_2Z_1-Z_1Z_2)f=[Z_2,Z_1]f=0.
$$
As $[V^X,V^X]=V^X$, every element $Z$ of $V^X$ can be written as the commutator of two elements of $V^X$ and, in view of the above expression,  $Zf=0$ which shows that $f$ is a $t$-independent
constant of motion for $X$. Finally, as $Y=-\widehat \Lambda(df)$, then $Y\in
\mathcal{S}^X_\Lambda$.
\end{proof}
Note that, roughly speaking, the above proposition ensures that, when $V^X$ is
{\it perfect}, i.e. $[V^X,V^X]=V^X$ (see \cite{Ca03}), then $\mathcal{S}_\Lambda^X$ contains all Hamiltonian Lie symmetries
of $X$. This is the case for Winternitz--Smorodinsky systems (\ref{LieS}) with a non-constant $\omega(t)$, whose
$V^X$ was already shown to be isomorphic to $\mathfrak{sl}(2,\mathbb{R})$.

\section{Conclusions and Outlook}
We have laid down the background for the analysis of a class of systems of
first-order ordinary differential equations whose dynamic is determined by a
curve in a Lie algebra of Hamiltonian vector fields, i.e. the Lie--Hamilton
systems. We proved that these systems can be described through curves in finite-dimensional real  Lie
algebras of functions on a Poisson manifold, the Lie--Hamiltonian structures.
Such structures have been employed to study features of Lie--Hamilton systems,
e.g. linearisability conditions, constants of motion, Lie symmetries,
superposition rules, etc. All our methods and results have been illustrated by
examples of mathematical and physical interest.

Apart from the results derived within this work, there remains a big deal of
further properties to be analysed: the existence of several Lie--Hamiltonian
structures for a system, the study of conditions for the existence of
Lie--Hamilton systems, methods to derive superposition rules, the analysis of
integrable and superintegrable Lie--Hamilton systems, etc. We plan to investigate all these
topics in the future.

\section*{Acknowledgements} 

Research of J.F. Cari\~nena and  J. de Lucas is partially financed by research projects MTM2009-11154 and MTM2010-12116-E (Ministerio de Ciencia e Innovaci\'on) and E24/1 (Gobierno de Arag\'on). J. de Lucas also acknowledges financial support by Gobierno de Arag\'on under project  FMI43/10
to accomplish a research stay in the University of Zaragoza. C. Sard\'on acknowledges a fellowship provided by the
University of Salamanca and partial financial support by research project FIS2009-07880 (Direcci\'on General de Investigaci\'on, Ciencia y Tecnolog\'ia). Finally, J. de Lucas would like to thank profs. Flores-Espinoza, Y. Vorobiev and G. D\'avila-Rasc\'on for their useful comments and valuable remarks.

\end{document}